\documentclass[12pt,fullpage]{article}
\usepackage[margin=1in]{geometry}
\usepackage{hyperref}
\usepackage{footnote}
\usepackage{amssymb}
\usepackage{amsmath,amsthm}
\usepackage{stmaryrd}
\SetSymbolFont{stmry}{bold}{U}{stmry}{m}{n} 
\usepackage{multicol, paracol, multirow}
\usepackage{color}
\usepackage{float}
\usepackage{todonotes}
\usepackage{bm}
\usepackage[latin1]{inputenc}

\usepackage{thmtools}
\usepackage{thm-restate}

\usepackage{cleveref}

\usepackage{amsfonts}
\usepackage{tcolorbox}
\usepackage{bbm}
\usepackage{complexity}
\usepackage{graphicx}
\usepackage{bussproofs}

\usepackage[linesnumbered,lined,boxed,commentsnumbered]{algorithm2e}
\usepackage{mathtools}
\usepackage{hyperref}
\usepackage{algpseudocode}
\newtheorem{definition}{Definition}

\newtheorem{theorem}{Theorem}

\newtheorem{lemma}{Lemma}

\usepackage{algorithmicx}

\newcommand{\PC}{\class{PC}}

\newcommand{\PCR}{\class{PCR}}

\newcommand{\Split}{\textrm{Split}}
\newcommand{\BOPOR}{BOP^{\vee}_{\ell,n}}
\newcommand{\rt}{R_{\tau(t)}(t)}

\algdef{SE}[DOWHILE]{Do}{doWhile}{\algorithmicdo}[1]{\algorithmicwhile\ #1}

\title{Polynomial Calculus sizes over the Boolean and Fourier bases are incomparable}
\author{Sasank Mouli \\ Mahindra University, Hyderabad \\ and \\ Indian Institute of Technology, Indore \\ sasankm.ucsd@gmail.com}
\date{}
\begin{document}

\maketitle

	\begin{abstract}
	For every $n >0$, we show the existence of a CNF tautology over $O(n^2)$ variables of width $O(\log n)$ such that it has a Polynomial Calculus Resolution refutation over $\{0,1\}$ variables of size $O(n^3\polylog(n))$ but any Polynomial Calculus refutation over $\{+1,-1\}$ variables requires size $2^{\Omega(n)}$. This shows that Polynomial Calculus sizes over the $\{0,1\}$ and $\{+1,-1\}$ bases are incomparable (since Tseitin tautologies show a separation in the other direction) and answers an open problem posed by Sokolov \cite{sok20} and Razborov \cite{razbopen}.
\end{abstract}

\section{Introduction}

Polynomial Calculus (PC), introduced in \cite{cei},  is a well-studied proof system with lower bounds known through a series of works \cite{razbpc}, \cite{ips99}, \cite{bgip}, \cite{alekh-raz}, \cite{galesi2010}, \cite{miksanord}. In spite of progress, these lower bounds have not shed light on limitations of systems such as $\AC^0[p]$-Frege, which was the original motive of studying Polynomial Calculus. A natural subsystem of $\AC^0[p]$-Frege for which lower bounds remained open is PC over $\{+1,-1\}$ variables or the ``Fourier'' basis (instead of the usual $\{0,1\}$ or Boolean basis considered in prior works). After this was highlighted in \cite{imp20}, Sokolov \cite{sok20} proved the first exponential lower bounds for PC over the Fourier basis, which were generalized in subsequent works \cite{imp23}, \cite{dmm23}.

Sokolov \cite{sok20} posed the natural problem of separating PC sizes over the Fourier and Boolean bases, which was also highlighted by Razborov \cite{razbopen}. Since Tseitin tautologies (or unsatisfiable $\mathbb{F}_2$-linear equations) have small Fourier proofs but require large Boolean proofs, he asked whether this is always the case, i.e. if any PCR proof of size $s$ of a CNF formula $F$ over the Boolean basis can always be converted to a proof of $F$  over the Fourier basis of size $s^{O(1)}$. In this work, we answer this question in the negative. We show a family of CNF formulas over $O(n^2)$ variables that has a PCR refutation of size $O(n^3\polylog(n))$, but any PC refutation over $\pm 1$ requires size $2^{\Omega(n)}$. 

The formulae we use are a variant of the Linear Ordering Principle ($LOP_n$), introduced by Krishnamurthy \cite{krish}, who conjectured that they require long Resolution proofs. St\aa lmarck \cite{stal} refuted this by demonstrating short Resolution proofs of these tautologies. They were then used by Bonet-Galesi \cite{bonet-galesi} and Galesi-Lauria \cite{galesi2010} to show optimality of size-width/degree tradeoffs for Resolution and Polynomial Calculus respectively. \cite{negreas} used an $OR$ lifted version of the graph variant of $LOP_n$ to show a separation between PC and PCR sizes. We use an $OR$ lifted version of $BOP_n$, which is obtained by shortening the clauses of $LOP_n$ by using binary extension variables for our purposes. We state our main result below.

\begin{theorem}
	For any $n >0$, there exists a CNF tautology over $O(n^2)$ variables of width $O(\log n)$ which has PCR proofs of size $O(n^3\polylog(n))$ over $\{0,1\}$ but requires PC proofs of size $2^{\Omega(n)}$ over $\{+1,-1\}$.
\end{theorem}

\section{Preliminaries}

\begin{definition}[Polynomial Calculus/Polynomial Calculus Resolution]
	
	\label{def:PC-PCR}
	Let $\Gamma = \{P_1 \ldots  P_m\}$ be an unsolvable system of polynomials in variables $ \{x_1 \ldots x_n\}$ over $\mathbb{F}$. A $\PC$ (\textit{Polynomial Calculus}) refutation of $\Gamma$ is a sequence of polynomials $\{R_1 \ldots  R_s\}$ 
	where $R_s=1$, and for every $\ell \in [s]$,
	$R_\ell$ is either a polynomial from $\Gamma$,
	or is obtained from two previous polynomials $R_j, R_k$, $j,k < \ell$, by one of the following derivation rules: 
	\begin{itemize}
		\item[] $R_\ell = \alpha R_j+\beta R_k$ for $\alpha$, $\beta$ $\in$ $\mathbb{F}$
		\item[] $R_\ell = x_iR_k$ for some $i \in [n] $
	\end{itemize}
	
	The size of the refutation is $\sum_{\ell=1}^s{|R_\ell|}$, where $|R_\ell|$ is the number of monomials in the polynomial $R_\ell$. The degree of the refutation is $\max_\ell deg(R_\ell)$.
	
	A $\PCR$ (\textit{Polynomial Calculus Resolution}) refutation is a $\PC$ refutation over the set of Boolean variables $\{x_1 \ldots x_n, \bar{x}_1 \ldots \bar{x}_n\}$ where $\{\bar{x}_1 \ldots \bar{x}_n\}$ are \textit{twin} variables of $\{x_1 \ldots x_n\}$.
	That is, over the $\{0,1\}$ encoding, the equations $x_i^2 - x_i = 0$, $\bar{x_i}^2 - \bar{x_i} = 0$ and $x_i+\bar{x}_i -1= 0$ are treated as axioms.
	Similarly, over the $\pm 1$ encoding, the equations $x_i^2 - 1 = 0$, $\bar{x_i}^2 -1 = 0$ and $x_i\bar{x}_i +1 = 0$ are treated as axioms.

\end{definition}

 \begin{definition}[Quadratic set, Quadratic degree, Quadratic terms over $\pm 1$; taken from \cite{sok20}, Section 3.2]\label{def:qset}
	Given a proof $\Pi$ over $\pm 1$ variables, the Quadratic set of $\Pi$, denoted $\mathcal{Q}(\Pi)$, is the set of pairs of terms $\mathcal{Q}(\Pi) = \{(t_1,t_2) ~|~ t_1,t_2 \in P \text{ for some line } P \in \Pi \}$.\\
	Denote by $\mathcal{QT}(\Pi)$ the set of quadratic terms $\{ t_1t_2 \mid (t_1,t_2)\in \mathcal{Q}(\Pi)\}$, where the product is modulo the axioms $x_i^2 = 1$. \\
	The Quadratic degree of $\Pi$ is the max degree of a term in $\mathcal{QT}(\Pi)$.\\
	Informally, Quadratic degree is the max degree of the square of each line (before cancellations). 
\end{definition}

	\begin{definition}[$\Split$ operation over $x$ \cite{sok20}, Section 5.4]
	\label{def:split-sok20}
	Given a proof $\Pi=(P_1,P_2,\ldots ,P_t)$ and a variable $x \in \{\pm 1\}$, each  line $P_i$ of $\Pi$ is of the form $P_{i,1}x+P_{i,0}$, where $P_{i,1},P_{i,0}$ do not contain $x$. The $\Split$ operation at $x$, denoted by $\Split_x(\Pi)$, is the sequence $\Pi'$ with the lines $\{P_{1,1},P_{1,0}, P_{2,1},P_{2,0}, \ldots , P_{t,1}, P_{t,0}\}$.
\end{definition}

The following lemmas show that applying $\Split_x$ on a proof $\Pi$ gives a valid proof when $x$ does not appear in the axioms of $\Pi$, and it removes from the proof all Quadratic terms containing $x$. They are adapted from \cite{sok20} and \cite{imp23}, with the below versions from \cite{dmm23}

	\begin{lemma}[\cite{sok20}] \label{split2}
	Suppose that $\Pi$ is a proof and $x$ is a variable that does not appear in any axioms of $\Pi$ except $x^2 = 1$. Then $\Split_x(\Pi)$ outputs a valid proof of the axioms of $\Pi$, with no line containing $x$.
\end{lemma}

	\begin{proof}
		Let $\Pi$ be the sequence $P_1,\ldots, P_t$. 
		We show  by induction on the line number $j$ that both $P_{j,1}$ and $P_{j,0}$ are derivable and $x$-free.
		
		If $P_j$ is an axiom, then it is free of $x$. So the $\Split$ version
		is $P_{j,1}=0$, $P_{j,0}=P_j$, and both these polynomials are derivable.          
		
		If $P_j = \alpha P_i + \beta P_k$ for some $i,k<j$, then
		$P_{j,b}=\alpha P_{i,b}+\beta P_{k,b}$ for $b=0,1$.
		
		If $P_j = y P_i$ for some $i<j$ and some variable $y\neq x$, then
		$P_{j,b}= y P_{i,b}$ for $b=0,1$.
		
		If $P_j = x P_i$ for some $i<j$, then since $x^2=1$ we obtain
		$P_{j,1}= P_{i,0}$ and $P_{j,0}= P_{i,1}$.
		
		Thus all the lines $P_{j,b}$ are derivable and do not contain $x$.
		
		Since the last line of the proof is $P_t=1$, we have $P_{t,1}=0$ and $P_{t,0}=P_t=1$. Thus $\Split_x(\Pi)$ derives 1 and is a valid proof from the axioms of $\Pi$.
	\end{proof}

\begin{lemma}[\cite{sok20}]
	\label{split}
	Let $\mathcal{Q}_x(\Pi)$ be the set of pairs  $(t_1,t_2) \in \mathcal{Q}(\Pi)$ such that $x\in t_1t_2$, and let $\mathcal{QT}_x(\Pi)$ be the corresponding set of quadratic terms. \\
	If $(t_1,t_2)\in \mathcal{Q}(\Split_x(\Pi))$, then $t_1$ and $t_2$ are both $x$-free, and at least one of $(t_1,t_2)$, $(t_1x,t_2x)$, is in $\mathcal{Q}(\Pi)$. Thus 
	$\mathcal{QT}(\Split_x(\Pi)) \subseteq \mathcal{QT}(\Pi) \setminus \mathcal{QT}_x(\Pi)$.
\end{lemma}

\begin{proof}
	Consider a pair $(t_1,t_2) \in \mathcal{Q}(\Split_x(\Pi))$. That $t_1,t_2$ are $x$-free follows from \cref{split2}. 
	The pair $(t_1,t_2)$ is contributed to $\mathcal{Q}(\Split_x(\Pi)$ by  $P_b$ for some line $P=xP_1+P_0$ of $\Pi$ and some $b\in\{0,1\}$. If $P_0$ contributes the pair, then $P$ also contributes the pair to $\mathcal{Q}(\Pi)$. If $P_1$ contributes the pair, then $P$ contributes the pair $(t_1x,t_2x)$ to $\mathcal{Q}(\Pi)$.  
\end{proof}

The following lemma from \cite{sok20} shows how to convert from low Quadratic degree to low degree. For completeness we include a proof of it from \cite{dmm23} in the Appendix.
\begin{lemma}[\cite{sok20}, Lemma 3.6] \label{qdegtodeg}
	Let $\Pi$ be a refutation of a set of axioms $F$ of degree $d_0$ with Quadratic degree at most $d$. Then there exists a refutation $\Pi'$ of $F$ with (usual) degree at most $2\max(d,d_0)$.
\end{lemma}

\section{Proof of the main theorem}

\subsection{The tautology}
 The tautology we use must be easy for PCR over $\{0,1\}$, but must allow us to prove a lower bound for PC over $\{+1,-1\}$. For this we use a short-width, $OR$-lifted version of the tautology \emph{Linear Ordering Principle} from \cite{bonet-galesi}, \cite{galesi2010}, \cite{negreas} which states that in any total ordering of $n$ elements, there is a minimum element. We describe it formally below.
 
 \begin{definition}[Linear Ordering Principle, $LOP_n$]
 	Suppose that we have a set $\mathcal{A}$ of $n$ elements. Let $x_{ij}$, $i,j \in [n], i \neq j$ denote Boolean variables encoding a total order on the elements in $\mathcal{A}$ where $x_{ij} = 1$ if and only if element $i$ is ordered before element $j$. Then the unsatisfiable formula $LOP_n$ is the following set of $O(n^3)$ clauses over $O(n^2)$ variables, which states that there exists an ordering with no minimum element.
 	
 	$$  \bigvee_{i \neq j} x_{ij} ~~~ \forall j $$
 	$$ \bar{x}_{ij} \vee \bar{x}_{jk} \vee x_{ik} ~~~ \forall i,j,k$$
 	$$\bar{x}_{ij} \vee \bar{x}_{ji} ~~~ \forall i,j$$

 \end{definition}

The first set of clauses asserts that every element $j$ must have some other element ordered before it. We call these the \emph{vertex} axioms, and denote the $j^{th}$ vertex axiom by $V_j$. The next two sets of clauses enforce the total order: transitivity and asymmetry must hold. We call these the \emph{ordering} axioms $T$.

The lemma below states that $LOP_n$ has small sized Resolution refutations. This is originally from \cite{stal}, \cite{bonet-galesi}, but we will use the version from \cite{negreas}

\begin{lemma}[\cite{negreas}, Proposition 7] \label{opshort}
	The formula $LOP_n$ has resolution refutations of size $O(n^3)$, where each clause in the refutation has at most two negative literals. 
\end{lemma}
 
 For our purposes we need that the above formula be transformed into one with short initial clauses.
 \begin{definition}[Ordering Principle with binary pointers, $BOP_n$]
 	The tautology $BOP_n$ over variables $x_{ij}$, $i \in [n]$, $j \in [n]$ and variables $y_{ja}$, $j \in [n]$, $a \in [\log n]$, is the following set of clauses, where the expression $(y_i \neq j)$ denotes the clause $\bigvee_{b_a = 0}y_{ia}\bigvee_{b_a = 1}\bar{y}_{ia}$, with $b_{\log n} \dots b_1$ being the binary representation of $j$.
 	
 	$$  (y_j \neq i) \vee x_{ij} ~~~ \forall j \in [n]$$
 	 	$$ \bar{x}_{ij} \vee \bar{x}_{jk} \vee x_{ik} ~~~ \forall i,j,k \in [n]$$
 	$$\bar{x}_{ij} \vee \bar{x}_{ji} ~~~ \forall i,j \in [n]$$

  \end{definition}
 
$BOP_n$ is obtained from $LOP_n$ by using extension variables $y_j$ to reduce the clause width of $LOP_n$. $y_j = i$ forces $x_{ij}$ to be true. Similar to $LOP_n$, we denote by $BV_j$ the set of all axioms $\{(y_j \neq i)\vee x_{ij} , i \in [n]\}$, and by $T$ the set of all ordering axioms. $BOP_n$ is of width $O(\log n)$ and has a similar number of variables and clauses as $LOP_n$. Below we define the $OR$ lifted version $BOP^{\vee}_{\ell, n}$ of $BOP_n$.

\begin{definition}
	For a constant $\ell > 0$, let $BOP^{\vee}_{\ell,n}$ denote the tautology obtained from $BOP_n$ by replacing each $x_{ij}$ by an $OR$ of $\ell$ new variables $x_{ij1} \dots x_{ij\ell}$. The following are clauses of $\BOPOR$ :
	
		$$  (y_j \neq i) \bigvee (\vee_l x_{ijl}) ~~~ \forall j \in [n]$$
	$$ \bar{x}_{ijl} \bigvee \bar{x}_{jkl} \bigvee (\vee_l x_{ikl}) ~~~ \forall i,j,k \in [n] ,l \in [\ell]$$
	$$\bar{x}_{ijl} \bigvee \bar{x}_{jil} ~~~ \forall i,j \in [n] ,l \in [\ell]$$
	
\end{definition}  

Let $BV_j^\vee$ and $T^\vee$ denote the $j^{th}$ vertex axioms and the ordering axioms respectively of $\BOPOR$. One of our contributions is a \emph{special degree}  lower bound for PC proofs of $\BOPOR$, where we define this notion carefully. This is along the lines of a similar bound proved in \cite{galesi2010} for the ordering principle based on expander graphs. We state the lower bound here and defer the proof to the next section. To state the bound, we need the following definitions of special degree.

\begin{definition}
	Let $t$ be a term in the variables of $\BOPOR$. We say that a vertex $j$ is strongly touched by $t$ if $x_{ijl} \in t$ for some $i,l$ or $y_{ja} \in t$ for some $a$. We say that a vertex $i$ is lightly touched (through $j$) if for some $j$ and for all $l$, $x_{ijl} \in t$. The set of all touched vertices, denoted by $\tau(t)$, is the set of all vertices that are both strongly and lightly touched by $t$. We sometimes refer to $|\tau(t)|$ as special degree.
\end{definition}

\begin{theorem}[Special PC degree lower bound for $\BOPOR$] \label{bopdeg}
	In any PC refutation of the tautology $\BOPOR$, there must exist a term $t$ such that $|\tau(t)| = n$.
\end{theorem}

On the other hand, since $  \forall i ~ (y_j \neq i) \vee x_{ij}$ implies $\bigvee_{i}x_{ij}$ we can derive $LOP_n$ from $BOP_n$ in size $O(n)$ (since the number of variables in $y_j$ is logarithmic). \cite{negreas} show that the $OR$ lifted version of $LOP_n$ also has small Resolution refutations.

\begin{lemma}[\cite{negreas}, Lemma 9]
	The formula $LOP_n^{\vee}$, obtained by lifting each variable of $LOP_n$ by an OR of $\ell$ variables, has a Resolution refutation of size $O(n^3\poly(\ell))$.
\end{lemma}

\begin{lemma}[Short Resolution refutations of $\BOPOR$]
	$\BOPOR$ has Resolution refutations of size $O(n^3\poly(\ell))$
\end{lemma}
\begin{proof}
	As stated above, $LOP_n$ can be derived from $BOP_n$ in Resolution size $O(n)$. Substituting $OR$ gadget variables for the vertex variables of $BOP_n$, we obtain a derivation of $LOP_n^\vee$ from $\BOPOR$. We now use the previous lemma.
\end{proof}

Part of our lower bound proof uses a random clustering argument on the gadget variables, where we pair variables $x_{ij1} \dots x_{ij\ell}$ randomly and assign each pair to a new variable $z_{ijl}$ to obtain a smaller set of gadget variables $z_{ij1} \dots z_{ij\ell/2}$. The following lemma shows that any term $t$ of degree at least $\ell/2$ in variables $x_{ijl}$ in the $\pm 1$ basis reduces to a term that misses at least one variable $z_{ijl}$ with high probability after random clustering.

\begin{lemma}\label{cluster}
	Suppose that $t$ is a term over $\{+1,-1\}$ in variables $x_{ij1} \dots x_{ij\ell}$ of degree at least $\ell/2$. We iteratively pick a pair of variables $x_{ijl_1}$ and $x_{ijl_2}$ at random and assign them to a new variable $z_{ijl}$, so that $t$ is now a term over variables $z_{ij1} \dots z_{ij\ell/2}$. Then with high probability for a large enough $\ell$, $t$ does not contain at least one variable from $z_{ij1} \dots z_{ij\ell/2}$.
\end{lemma} 
\begin{proof}
	We observe that if the random process pairs any two variables $x_{ijl_1}, x_{ijl_2} \in t$, then the corresponding new variable $z_{ijl}$ does not appear in $t$ after substitution. Therefore, $t$ contains all the variables $z_{ijl}$ only if every variable in $t$ is paired with a variable outside of $t$. The probability that this happens is at most $(3/4)^{\ell/2}$.
\end{proof}

We are ready to prove our size lower bound for PC refutations of $\BOPOR$ over $\pm 1$.

\begin{theorem}[Main theorem]
	For any $n > 0$ and $\ell > 10 \log n$, any refutation of $\BOPOR$ in PC over $\pm 1$ requires size $2^{\Omega(n)}$
\end{theorem}
\begin{proof}
	Let $\Pi$ be a PC refutation of $\BOPOR$ over $\pm 1$ of size $2^{\gamma n}$ for a small enough constant $\gamma > 0$. We will reduce $\Pi$ to a refutation $\Pi'$ of $BOP^{\vee}_{\ell/2,n/2}$ of special degree less than $n/2$, contradicting Theorem \ref{bopdeg}. We consider two parts of special degree of any term $t \in \Pi$: the number of strongly touched vertices and the number of lightly touched vertices. We lower both of these separately to arrive at our lower bound. For the former, we consider Quadratic terms of the proof that strongly touch at least $n/8$ vertices. By iteratively picking a vertex $j$ that is strongly touched by many Quadratic terms and applying $\Split$ on its variables, we remove all Quadratic terms with many strongly touched vertices. We then use Lemma \ref{qdegtodeg} to get a refutation where each term of the proof strongly touches less than $n/4$  vertices. For the latter, we randomly cluster our gadget variables as in Lemma \ref{cluster}, and argue that after clustering no term lightly touches more than $n/4$ vertices with high probability. This gives us a refutation of special degree less than $n/2$.
	
	Let $H$ be the set of all Quadratic terms in $\Pi$ which strongly touch at least $n/8$ vertices. Pick a vertex $j$ that by averaging is strongly touched by at least ${1/8}^{th}$ of Quadratic terms in $H$. This means that an eighth of the terms in $H$ contain $x_{ijl}$ for some $i,l$ or contain $y_{ja}$ for some $a$. We pick the gadget indices $l_i$ such that $x_{ijl_i}$ occurs the most in terms containing $x_{ijl}$, $l \in [\ell]$. For every $i$ and $l \neq l_i$, we set $x_{ijl} = 1$. We also set $x_{jkl} = 0$ for all $k,l$. Note that this eliminates all axioms in $BV_j^\vee$ plus any axiom in $T^\vee$ that mentions $j$ and we recover copy of $BOP^{\vee}_{\ell,n -1}$ where each $BV_k^\vee$ contains an additional axiom of the form $(y_k \neq j)$. Functionally this has no effect and the special degree lower bound still holds for this copy. We now apply $\Split$ on all variables $y_{ja}$ and on all $x_{ijl_i}$. Note that since none of these variables are in the axioms, this gives a valid refutation by Lemma \ref{split2}. Moreover, this eliminates from the refutation at least a $1/8\ell$ fraction of terms in $H$ by Lemma \ref{split}. By repeating this for $8\ell \gamma n$ iterations, we remove all Quadratic terms in $H$. By using Lemma \ref{qdegtodeg}, we obtain a refutation of $BOP^{\vee}_{\ell,n/2}$ where each term strongly touches less than $n/4$ vertices (for a small enough $\gamma$).
	
	We now cluster gadget variables $x_{ijl}, l \in [\ell]$ into variables $z_{ijl}, l \in [\ell/2]$ as in Lemma \ref{cluster}. We compute the probability that a term $t$ lightly touches more than $n/4$ vertices after clustering. For any index $i$ and some $j$, if $t$ contained less than $\ell/2$ variables $x_{ijl}$ prior to clustering, then it still contains less than $\ell/2$ variables after clustering, and thus does not lightly touch $i$ through $j$. On the other hand, if $t$ contained at least $\ell/2$ variables $x_{ijl}$, then by Lemma \ref{cluster} it now misses some variable $z_{ijl}$ with probability all but $(3/4)^{\ell/2} < 1/n^2$. Therefore, by a union bound, the probability that $i$ is lightly touched by $t$ through some $j$ is at most $1/n$ and thus the probability that $t$ lightly touches more than $n/4$ vertices is at most ${{n/2}\choose{n/4}} (1/n)^{n/4} < 2^{-O(n\log n)}$. Since the size of the proof is at most $2^{\gamma n}$, again by a union bound we have that there exists a clustering where each term lightly touches at most $n/4$ vertices.
	
	Therefore, we have a refutation $\Pi'$ of $BOP^{\vee}_{\ell/2,n/2}$ where the number of touched vertices by any term is less than $n/2$. This contradicts Lemma \ref{bopdeg}.

\end{proof}	

\section{Proof of the special degree lower bound}

In this section we prove the lower bound for PC stated in Theorem \ref{bopdeg}. For this we largely follow the framework laid out in \cite{galesi2010}, which in turn was motivated by \cite{alekh-raz}.

\subsection{Preliminaries}
We work with a field $\mathbb{F}$ and polynomials in the ring $\mathbb{F}[x_1 ,\dots, x_n]$. 

\begin{definition}[Ideal generated by a set of polynomials $F$]
	Given a set of polynomials $F = \{f_1 \dots f_k\}$, the ideal generated by $F$, denoted by $Span(F)$ is the set of all polynomials of the form $\big\{\sum_i{f_ig_i} ~|~ g_i \in \mathbb{F}[x_1 ,\dots, x_n] \big\}$
\end{definition}

\begin{definition}[Semantic implication]
We say that a set of polynomials $f_1, \dots, f_n$ semantically implies another polynomial $g$, if for every assignment to the underlying variables, $g = 0$ whenever each $f_i = 0$. We denote this by $f_1, \dots, f_n \vDash g$.
\end{definition}

\begin{definition}[Graded lexicographic ordering]
	Let $\prec$ denote the following ordering on multilinear monomials in $\mathbb{F}[x_1, \dots x_n]$. We have $1 \prec x_1 \prec x_2 \prec  \dots \prec x_n$. For any two monomials $M$,$M'$, if $\deg(M) < \deg(M')$, then $M \prec M'$. Also, if $M \preceq M'$ then for any variable $x$, we have $xM \preceq xM'$ (where the product is modulo $x^2 = x$). Note that this is a total ordering on multilinear monomials.
	
	For a polynomial $P$, the leading term, denoted by $LT(P)$ is the greatest monomial in $P$ according to $\prec$. This ordering is extended lexicographically to polynomials, i.e. for two polynomials $P = a_1M_1 + a_2M_2 + \dots + a_kM_k$ and $P' = a'_1M_1' + a'_2M_2 + \dots + a'_kM'_k$, we write $P \prec P'$ if there exists an $i \in [k]$ such that $M_j = M'_j$ for $j < i$ and $M_i \prec M_i'$. If no such $i$ exists then $P$ and $P'$ are incomparable w.r.t. $\prec$ and therefore $\prec$ is not a total order on polynomials. However, in such a case we have that a linear combination of $P$ and $P'$ is less than both $P$ and $P'$ according to $\prec$. Therefore, in any set of polynomials there always exists a minimum element w.r.t. $\prec$.
\end{definition}
	
\begin{definition}[Residue of a polynomial $P$ w.r.t. a set of polynomials $F$]
	Given a set of polynomials $F$, the residue $R_F(P)$ of a polynomial $P$ modulo the ideal generated by $F$ is defined as the minimum element in the set $\{P - Q ~|~ Q \in Span(F)\}$ according to $\prec$ (which exists by the previous definition).
\end{definition}

\begin{definition}[Properties of residue, cf. \cite{galesi2010}, Property 1] \label{rprop}
	For a set of polynomials $F$ and any two polynomials $P$, $Q$, we have:
	
	\begin{enumerate}
		\item $R_F(P) \preceq P$
		\item If $P-Q \in Span(F)$, $R_F(P) = R_F(Q)$
		\item $R_F$ is a linear operator
		\item $R_F(PQ) = R_F(P \cdot R_F(Q))$
	\end{enumerate}
	 
\end{definition}

\subsection{Proof of the special degree lower bound}

Our lower bound proof relies on building a linear operator $R$ with the following properties, which is sufficient for proving the desired lower bound from Theorem \ref{bopdeg}.

\begin{lemma}[\cite{galesi2010}, \cite{alekh-raz}] \label{rop}
	Suppose that there exists a linear operator $R$ such that 
	\begin{enumerate}
		\item $R(A) = 0$ for every axiom $A$ of $BOP_n$
		\item For any polynomial $P$ and any variable $w$ such that $|\tau(wt)| < n$ for any $t \in P$, $R(wP) = R(w\cdot R(P))$ 
		\item $R(1) = 1$
	\end{enumerate}
	
	Then $BOP_n$ requires proofs of where there exists a term $t$ with $|\tau(t)| = n$ .
\end{lemma}

We conclude by showing the existence of such a linear operator $R$. Again, we follow the framework of \cite{galesi2010}, \cite{alekh-raz} and choose $R(t)$ for a term $t$ to be the residue $R_{\tau(t),T^\vee}(t)$, where we abuse notation to indicate by $\tau(t)$ the set of all axioms $BV_j^\vee$ for $j \in \tau(t)$. We drop the $T^\vee$ and write $R_{\tau(t)}(t)$ to mean $R_{\tau(t),T^\vee}(t)$.



We would like to show that for any term $t$ with $|\tau(t)|<n-1$ and a variable $w$, $\tau(wt)$ cannot do any better in reducing $t$ than $\tau(t)$. This is made formal below.

\begin{lemma}\label{rdrop}
	Let $t$ be term and let $w$ be any variable such that $|\tau(wt)|<n$. Then we have $R_{\tau(wt)}(t) = R_{\tau(t)}(t)$.
\end{lemma}
\begin{proof}
	We have by definition of $R_{\tau(wt)}(t)$ that $t - R_{\tau(wt)}(t)$ belongs to $Span(\tau(wt))$, or in other words we have
	
	$$ T^\vee, \tau(wt) \vDash t - R_{\tau(wt)}(t) $$
	
	Since we have $|\tau(wt)| \leq n-1$ there exists a vertex axiom  $BV_j^\vee \notin \tau(wt)$. This means that for every $i,l$ $x_{ijl} \not \in t$, and for every $k$ there exists some $l_k$ such that $x_{jkl_k} \not \in t$. Consider the following partial assignment $\rho$ to the variables of $\BOPOR$ (called a $j$-cta in \cite{galesi2010}): set $x_{ijl} = 0$ for all $i,l$ and set $x_{jkl_k} = 1$ for every $k$. Also set $y_{k} = j$ for $BV_k^\vee \in \tau(wt)\setminus \tau(t)$ and leave the other variables unset. In other words, $\rho$ sets the element $j$ as the minimum, sets the pointers $y_k$ of the vertices $k$ touched by $w$ and not touched by $t$ (if they exist) to point to the vertex $j$ and leaves everything else unset. We note here that $\rho$ is chosen such that none of the variables in $t$ are set, and each axiom in $\tau(wt)\setminus \tau(t)$ is killed.
	
	Let us see the effect $\rho$ has on the above semantic implication. Firstly, we claim that $T_{|_\rho}^\vee \subseteq T^\vee$. This is easy to see: $\rho$ kills all the ordering axioms involving $j$ and does not touch the other ordering axioms. Similarly, we claim that the set of axioms in $\tau(wt)_{|_\rho}$ is a subset of those in $\tau(t)$. This is because as mentioned above, every axiom in $\tau(wt)\setminus \tau(t)$ is killed by $\rho$. Moreover, for any $BV_k \in \tau(t)$, the axiom pointing to $j$ is killed and others are left untouched. Therefore we have the semantic implication
	
	$$ T^\vee, \tau(t) \vDash t - {R_{\tau(wt)}(t)}_{|_\rho} $$
	
	This is because $t$ does not contain any variable that is set by $\rho$. We now have by minimality of $R_{\tau(t)}(t)$ that $R_{\tau(t)}(t) \preceq {R_{\tau(wt)}(t)}_{|_\rho}$. However, since a restriction can only decrease the degree of any term, we have that ${R_{\tau(wt)}(t)}_{|_\rho} \preceq {R_{\tau(wt)}(t)}$. Finally since $\tau(wt) \supseteq \tau(t)$, we have ${R_{\tau(wt)}(t)} \preceq {R_{\tau(t)}(t)}$. We therefore have ${R_{\tau(t)}(t)} \preceq {R_{\tau(wt)}(t)} \preceq {R_{\tau(t)}(t)}$ and thus ${R_{\tau(t)}(t)} = {R_{\tau(wt)}(t)}$.
	
\end{proof}

We note that the above lemma also works when we replace $\tau(wt)$ by any set of the form $BV_I^\vee$, where $BV_I^\vee = \{BV_i^\vee ~|~ i \in I\}$, such that $|I| < n$ and $I$ contains $\tau(t)$.

\begin{lemma} \label{rdrop2}
	Let $I$ be any set of vertices of size less than $n$ and let $t$ be a term such that $\tau(t) \subseteq BV_I$. Then we have $R_{BV_I,T^\vee}(t) = \rt$.     
\end{lemma}

The following technical lemma is needed to make the proof of Lemma \ref{rop} go through.

\begin{lemma}\label{rtech}
	 Let $t$ a term and let $t'$ be any term appearing in $R_{\tau(t)}(t)$. Then we have $\tau(t') \subseteq \tau(t)$
\end{lemma}
\begin{proof}
	 The proof of this lemma is similar to the previous one. Suppose that there exists some term $t' \in R_{\tau(t)}(t)$ with $j \in \tau(t') \setminus \tau(t)$. Then we have 
	 
	 $$T^\vee, \tau(t) \vDash t - \rt $$
	 
	 Let $\rho$ be the partial assignment defined in the previous lemma where we additionally set $y_j$ to an arbitrary value. Substituting this in the above semantic implication, we have similar to the previous lemma
	 
	 $$T^\vee, \tau(t) \vDash t - \rt_{|_\rho} $$
	 
	 We now have $\rt_{|_\rho} \prec \rt$ since $t' \in \rt$ is hit by $\rho$ and its degree is now lesser. However from the above semantic implication we have $\rt \preceq \rt_{|_\rho}$, which is a contradiction.
\end{proof}

We conclude this section by showing the existence of a linear operator $R$ from Lemma \ref{rop}. As mentioned earlier, we define $R(t)$ to be $\rt$.

\begin{lemma}
	$R(t) = \rt$ satisfies the properties of Lemma \ref{rop} and therefore $\BOPOR$ requires PC proofs where $|\tau(t)| = n$ for some term $t$
\end{lemma}
\begin{proof}

We verify that our choice of $R$ satisfies each property listed in Lemma \ref{rop}.

\paragraph{1.} Let $A$ be an axiom of $\BOPOR$. If $A \in T^\vee$ is one of the ordering axioms, $R(A) = 0$ since we reduce by $T^\vee$ by default. On the other hand, if $A$ is an axiom in $BV_j^\vee$, then for any term $t \in A$, $\tau(t)$ contains $BV_j^\vee$ and thus $R(A)$ is again equal to zero.

\paragraph{2.} Let $P = \sum_i a_i t_i$ be a polynomial and let $w$ be any variable with $|\tau(wt_i)| < n$. Then we have  $R(wP) = \sum_i a_iR(wt_i) =\sum_i a_i R_{\tau(wt_i)}(wt_i)$. Now, by Definition \ref{rprop}, pt. 4 we have $R_{\tau(wt_i)}(wt_i) = R_{\tau(wt_i)}(wR_{\tau(wt_i)}(t_i))$. Since $|\tau(wt_i)| < n$, by Lemma \ref{rdrop} we have $R_{\tau(wt_i)}(t_i) = R_{\tau(t_i)}(t_i) =  R(t_i)$. Therefore we have $R(wP) = R_{\tau(wt_i)}(wR(t_i))$. We argue that the latter expression is in fact equal to $R(wR(t_i))$. To see this, consider any term $t \in wR(t_i)$. By Lemma \ref{rtech}, we have $\tau(t) \subseteq \tau(wt_i)$. Moreover, by Lemma \ref{rdrop2}, we have $R(t) = R_{\tau(t)}(t) = R_{\tau(wt_i)}(t)$ (since $|\tau(wt_i)| < n$). Putting everything together, we get that $R(wP) = R_{\tau(wt_i)}(wR(t_i)) = R(wR(t_i))$.

\paragraph{3.} We have that $\tau(1) = \emptyset$ and since $T^\vee$ is satisfiable, 1 is irreducible modulo $T^\vee$ and thus $R_{\tau(1)}(1) = 1$.

\end{proof}

\section*{Acknowledgments}

The author would like to thank Dmitry Sokolov and Alexander Razborov for posing this open problem, Russell Impagliazzo for pointing to the random clustering technique, Arkadev Chattopadhyay and Yogesh Dahiya for helpful discussions.

\appendix
\section{Proof of Sokolov's low Quadratic degree to low degree lemma}
	We include here a proof from \cite{dmm23} of Lemma \ref{qdegtodeg} . We restate it below.
	
	\begin{lemma}
		Let $\Pi$ be a refutation of a set of axioms $F$ of degree $d_0$ with Quadratic degree at most $d$. Then there exists a refutation $\Pi'$ of $F$ with (usual) degree at most $2\max(d,d_0)$.
	\end{lemma}
	
	\begin{proof}
		Let $\Pi = \{P_j\}_j$. Now, consider $\Pi'=\{P_j'\}_j$ where $P_j' = t_jP_j$, with each $t_j \in P_j$ carefully selected. Since the degree of $t_jP_j$ is bounded by the Quadratic degree of $P_j$, every line in $\Pi'$ is of degree at most $d$. However, $\Pi'$ is not an immediate valid refutation of $F$, but it can be transformed into one. We will show that each line of $\Pi'$ can be derived from previous lines and axioms of $F$ in degree at most 2$\max(d,d_0)$, completing the proof. We proceed by induction on line number $j$.
		
		If $P_j$ is an axiom, then we set $t_j$ to be an arbitrary term in $P_j$ and derive $P_j' = t_jP_j$ in degree $2d_0$ starting from $P_j$. \\
		(Note that in \cite{sok20}, it is  claimed that this step can be derived in degree $d_0$. But this is not always so.
		For instance,  if $p=x_1x_2x_3+x_2x_3x_4+x_3x_4x_1+x_4x_1x_2$ has degree $d_0=3$, and $d=2$, then for any term $t\in p$, $tp$ has degree $2$ but needs degree $4>\max\{d,d_0\}$ for the derivation.
		)
		
		If $P_i = xP_j$ for some $j<i$, then we select $t_i = xt_j$, and consequently, $P_i' = t_iP_i = t_jP_j = P_j'$ is derived without raising the degree.
		
		Finally, if $P_i = P_{j_1} + P_{j_2}$, we choose $t_i$ to be an arbitrary term in $P_i$ and derive $P_i' = t_iP_i = t_it_{j_1}P_{j_1}'+t_it_{j_2}P_{j_2}'$. We argue that the degree of both $t_it_{j_1}$ and $t_it_{j_2}$ is at most $d$, and as a result, $P_i'$ can be derived from $P_{j_1}'$ and $P_{j_2}'$ in degree at most $2d$, which completes the proof. To justify this assertion, let $t_i \in P_{j_1}$ without loss of generality (every term in $P_i$ appears in either $P_{j_1}$ or $P_{j_2}$). Then degree of $t_it_{j_1}$ is bounded by the Quadratic degree of $P_{j_1}$ and hence by $d$. Additionally, if $t_{j_2} \in P_{i}$, then the degree of $t_it_{j_2}$ is bounded by the quadratic degree of $P_{i}$ and is also bounded by $d$. In the case where $t_{j_2} \not \in P_{i}$, it means that it was cancelled in the sum and therefore $t_{j_2} \in P_{j_1}$ and so degree of $t_it_{j_2}$ is bounded by the Quadratic degree of $P_{j_1}$ and is again bounded by $d$.
		
		Thus all lines in $\Pi'$ can be derived from previous lines and axioms of $F$ in degree at most 2$\max(d,d_0)$. Since the last line of $\Pi'$ is $1$, we get that $\Pi'$ can be successfully transformed into a valid proof of $F$ of degree 2$\max(d,d_0)$.
	\end{proof}

\bibliography{document.bib}
\bibliographystyle{alpha}

\end{document}